\documentclass[a4paper,11pt]{article}
\usepackage{amssymb,amsmath,amsthm,amsfonts}
\usepackage{xspace}
\usepackage[ruled,vlined,linesnumbered]{algorithm2e}
\usepackage{mathrsfs}  
\usepackage{cite}
\usepackage{thmtools}
\usepackage{thm-restate}

 \usepackage[pdftex, plainpages = false, pdfpagelabels, 
                 bookmarks=false,
                 bookmarksopen = true,
                 bookmarksnumbered = true,
                 breaklinks = true,
                 linktocpage,
                 pagebackref,
                 colorlinks = true,  
                 linkcolor = blue,
                 urlcolor  = blue,
                 citecolor = red,
                 anchorcolor = green,
                 hyperindex = true,
                 hyperfigures
                 ]{hyperref}

\usepackage[obeyFinal,colorinlistoftodos]{todonotes}

\newcommand{\pname}{\textsc}
\newcommand{\ProblemFormat}[1]{\pname{#1}}
\newcommand{\ProblemIndex}[1]{\index{problem!\ProblemFormat{#1}}}
\newcommand{\ProblemName}[1]{\ProblemFormat{#1}\ProblemIndex{#1}{}\xspace}

 \newcommand{\bmfgfr}{\ProblemName{Low GF(2)-Rank Approximation}}

\newcommand{\rpartitioncenterstar}{{\sc Binary Constrained Partition Center$^{\star}$}\xspace}
\newcommand{\rpartitioncenter}{{\sc Binary Constrained Partition Center}\xspace}
\newcommand{\rcenter}{{\sc Binary Constrained $k$-Center}\xspace}
\newcommand{\rclustering}{{\sc Binary Constrained Clustering}\xspace}

\newcommand{\GF}{{GF}(2)\xspace}

\newcommand{\rank}{{\rm rank}\xspace}
\newcommand{\GFrank}{{\rm{GF}}(2)\text{{\rm -rank}}\xspace}
\newcommand{\hdist}{d_H}
\newcommand{\Hdist}{{\operatorname{cost}}}
\newcommand{\opt}{{\sf OPT}}
 \newcommand{\LRAGF}{\textsc{$\ell_1$-Rank-$r$ Approximation over \GF}\xspace}
\newcommand{\LRAB}{\textsc{Boolean $\ell_1$-Rank-$r$ Approximation}\xspace}

\newcommand{\bfA}{\mathbf{A}} 
\newcommand{\bfB}{\mathbf{B}}

\newcommand{\bfU}{\mathbf{U}}
\newcommand{\bfV}{\mathbf{V}}  

\newcommand{\bfa}{\mathbf{a}} 
 
\newcommand{\bfc}{\mathbf{c}} 
 
\newcommand{\bfe}{\mathbf{e}}
 
\newcommand{\bfs}{\mathbf{s}}

\newcommand{\bfx}{\mathbf{x}}

\newcommand{\bfy}{\mathbf{y}}

\newcommand{\cR}{\mathcal{R}}

\newtheorem{definition}{Definition}
\newtheorem{theorem}{Theorem}
\newtheorem{claim}{Claim}
\newtheorem{corollary}{Corollary}
\newtheorem{lemma}{Lemma}

\newtheorem{proposition}[theorem]{Proposition}

\newcommand{\OO}{{\cal O}}

\newcommand{\RR}{{\cal R}}

\newcommand{\ar}[2]{r}

\newcommand{\defproblem}[3]{
  \vspace{3mm}
\noindent\fbox{
  \begin{minipage}{.95\textwidth}
  \begin{tabular*}{\textwidth}{@{\extracolsep{\fill}}lr} \textsc{#1} \\ \end{tabular*}
  {\bf{Input:}} #2  \\
  {\bf{Task:}} #3
  \end{minipage}
  }
  \vspace{2mm}
}

\newcommand{\expect}{\mathbf{E}}
\newcommand{\expectation}[3][0]{%
   \ifcase#1 
        \expect [ #2 \mid #3 ] 
      \or   \expect\bigl[ #2 \bigm\vert #3 \bigr]
      \or   \expect \Bigl[ #2 \Bigm\vert #3 \Bigr] 
      \or \expect\biggl[ #2 \biggm\vert #3 \biggr]
      \or  \expect \Biggl[#2 \Biggm\vert #3 \Biggr]
   \else 
       \expect\left[ #2  \;\middle\vert\; #3 \right]
   \fi}

\newcommand{\probability}[3][0]{%
   \ifcase#1 
    \Pr( #2 \mid #3 ) 
      \or \Pr \bigl( #2 \bigm\vert #3 \bigr)
      \or \Pr \Bigl( #2 \Bigm\vert #3 \Bigr) 
      \or \Pr \biggl( #2 \biggm\vert #3 \biggr)
      \or \Pr \Biggl( #2 \Biggm\vert #3 \Biggr) 
   \else 
      \Pr \left( #2  \;\middle\vert\; #3 \right)
   \fi}

\usepackage{vmargin}
\setmarginsrb{1in}{1in}{1in}{1in}{0pt}{0pt}{0pt}{6mm}

\title{Low-rank binary matrix approximation in column-sum  norm
}

\author{
Fedor V. Fomin
\and
Petr A. Golovach
\and 
Fahad Panolan
\and 
Kirill Simonov
}

\date{Department of Informatics, University of Bergen, Norway\\ \texttt{\{fomin|petr.golovach|fahad.panolan|kirill.simonov\}@ii.uib.no}}

\begin{document}
\maketitle

\thispagestyle{empty}

\begin{abstract} 
We consider \LRAGF, where 
for     a binary $m\times n$ matrix $\bfA$ and a positive integer $r$, one seeks a binary  matrix $\bfB$ of rank at most $r$, minimizing 
the column-sum norm $\| \bfA -\bfB\|_1$.
 We show that for 
every $\varepsilon\in (0, 1)$, there is a {randomized} $(1+\varepsilon)$-approximation algorithm for     \LRAGF of running time $m^{\OO(1)}n^{\OO(2^{4r}\cdot \varepsilon^{-4})}$. 
This is the first  polynomial time approximation scheme (PTAS) for this problem. 


\end{abstract}

\setcounter{page}{1}

\section{Introduction}\label{sec:intro}
Low-rank matrix approximation is the method of compressing a matrix by reducing its dimension. It is the basic component of various  methods in data analysis including Principal Component Analysis (PCA), 
one of the most popular and successful techniques used for dimension reduction  in data analysis and machine learning \cite{pearson1901liii,hotelling1933analysis,eckart1936approximation}. In low-rank matrix approximation 
one seeks the best   low-rank approximation of data matrix $\bfA$ with matrix $\bfB$ solving
\begin{eqnarray}\label{eq_PCA}
\text{ minimize } \|\bfA-\bfB\|_\nu  \\ 
\text{ subject to } \rank(\bfB) \leq  r. \nonumber
\end{eqnarray}
Here $\| \cdot \|_\nu$ is some matrix norm. The most popular matrix norms studied in the literature are the 
  \emph{Frobenius} 
$||\bfA||_F^2 = \sum_{i, j} a_{ij}^2$ 
and the  \emph{spectral}  $\|\bfA\| _2=\sup_{x\neq 0}\frac{\|\bfA\bfx\| _{2}}{\|\bfx\| _2}$ norms.  
By the Eckart-Young-Mirsky theorem \cite{eckart1936approximation,MR0114821},  \eqref{eq_PCA} is efficiently solvable via
  Singular Value Decomposition (SVD) for these two norms.
  The spectral  norm is an “extremal” norm---it measures the worst-case stretch of the matrix. On the other hand, the   Frobenius norm  is  “averaging”.
Spectral norm is usually applied in the situation when one is interested in actual columns for the subspaces they define and is of greater interest in scientific computing and numerical linear algebra. The  Frobenius norm is widely used in statistics and machine learning, see the survey of Mahony~\cite{MahonyM11} for further discussions.


Recently there has been considerable interest in developing algorithms for low-rank matrix approximation problems  for binary (categorical) data.
 Such variants of dimension reduction for high-dimensional data sets with binary attributes
arise naturally in applications involving binary data sets,  like latent  semantic analysis \cite{berry1995using}, pattern discovery for gene expression\cite{Shen2009}, or web search models \cite{DBLP:journals/jacm/Kleinberg99}, see  \cite{DanHJWZ15,Jiang2014,GutchGYT12,Koyuturk2003,PainskyRF16,Yeredor11} for other applications.
In many such applications it is much more desirable to approximate a binary matrix $\bfA$ with a binary matrix $\bfB$ of small (\GF or Boolean) rank because it
could provide a deeper insight into 
the semantics associated with the original matrix. There is a big body of work done on binary and Boolean low-rank matrix approximation, see   \cite{Bartl2010,BelohlavekV10,DanHJWZ15,LuVAH12,MiettinenMGDM08,DBLP:conf/kdd/MiettinenV11,Mitra:2016,VaidyaAG07,DBLP:conf/icde/Vaidya12} for further discussions.

Unfortunately, SVD is not applicable for the binary case  which makes such problems computationally much more challenging.  For binary matrix, its Frobenius norm  is equal to the number of its $1$-entries, that is  
$\|\bfA\| _F= \sum_{j=1}^n \sum_{i=1}^m | a_{ij} |$. Thus, the value $\|\bfA\ -\bfB\| _F$ measures the total Hamming distance from points (columns) of $\bfA$ to the subspace spanned by the columns of $\bfB$. For this variant of the low-rank binary matrix approximation, a number of approximation algorithms were developed, resulting in efficient  polynomial time approximation schemes (EPTASes)  obtained in 
\cite{BanBBKLW19, DBLP:journals/corr/abs-1807-07156}.
However, the algorithmic complexity of the problem for any vector-induced norm, including the spectral norm, remained open.
 
For binary matrices, the natural ``extremal'' norm to consider is the
$\| \cdot\|_1$ norm, also known as
  \emph{column-sum norm},  
  operator $1$-norm,  or H\"older matrix 1-norm. That is,  
for a matrix $\bfA$, 
\[\|\bfA\| _1=\sup_{x\neq 0}\frac{\|\bfA\bfx\| _{1}}{\|\bfx\| _1} = \max_{1 \leq j \leq n} \sum_{i=1}^m | a_{ij} |.\]
In other words, the column-sum norm is the maximum number of $1$-entries in a column in $\bfA$, whereas the Frobenius norm is the total number of $1$-entries in $\bfA$.
The column-sum norm is analogous to the spectral norm, only it is induced by the $\ell_1$ vector norm, not the $\ell_2$ vector norm.

We consider the  problem, where for  an $m\times n$
 binary data matrix $\bfA$ and a positive integer $r$, one seeks a binary matrix $\bfB$ optimizing 
\begin{eqnarray}\label{eq_PCA_1}
\text{ minimize } \|\bfA-\bfB\|_1  \\ 
\text{ subject to } \rank(\bfB) \leq  r. \nonumber
\end{eqnarray}
Here, by the rank of the binary matrix $\bfB$ we mean  its  \GF-rank.
%
 We refer to the problem defined by \eqref{eq_PCA_1}  as to \LRAGF.
The value  $\|\bfA\ -\bfB\| _1$ is the maximum  Hamming distance from each of the columns of $\bfA$ to the  subspace spanned by columns of $\bfB$ and thus, compared to approximation with the Frobenius norm, it  could provide a more accurate dimension reduction.    

It is easy to see by the reduction from the \textsc{Closest String} problem,  that already for $r=1$,  \LRAGF is NP-hard. 
The  main result of this paper is that \eqref{eq_PCA_1} admits a polynomial time approximation scheme (PTAS). 
More precisely, we prove the following theorem. 

\begin{theorem}
\label{thm:norm1PTAS} For every $\varepsilon\in (0, 1)$, there is a {randomized} $(1+\varepsilon)$-approximation algorithm for     \LRAGF of running time $m^{\OO(1)}n^{\OO(2^{4r}\cdot \varepsilon^{-4})}$.
\end{theorem}

In order to prove Theorem~\ref{thm:norm1PTAS} we obtain a PTAS for a more general problem, namely \rcenter. This problem has a strong expressive power and can be used to obtain PTASes for a number of problems related to  \LRAGF. For example, for the variant, when the rank of the matrix $\bfB$ is not over \GF but is Boolean. Or a variant of clustering, where we want to partition binary vectors into groups, minimizing the maximum distance in each of the group  to some subspace of small dimension. We provide  discussions of other applications of our work in Section~\ref{sec:applications}.

\paragraph*{Related work.}  
 The variant of \eqref{eq_PCA} with both matrices $\bfA$ and $\bfB$ binary, and $\|\cdot\|_{\nu}$ being the Frobenius norm, is known as \bmfgfr. 
Due to  numerous applications,
  various heuristic algorithms for \bmfgfr  could be found in the literature
 \cite{DBLP:conf/icdm/JiangH13a,Jiang2014,fu2010binary,Koyuturk2003,Shen2009}. 

When it concerns  rigorous algorithmic analysis of  \bmfgfr, 
Gillis and Vavasis ~\cite{GillisV15} and Dan et al. \cite{DanHJWZ15} have shown that \bmfgfr is NP-complete  for every $r\geq1$. A subset of the authors studied parameterized algorithms for \bmfgfr in \cite{FominGP18}. The first approximation algorithm for \bmfgfr is due to 
Shen et al. \cite{Shen2009}, who gave a $2$-approximation algorithm for the special case of $r=1$. For rank $r>1$,  Dan et al. \cite{DanHJWZ15} have shown that a $(r/2 +1 +\frac{r}{2(2^r-1)})$-approximate solution can be formed from $r$ columns of the input matrix $\bfA$. Recently, these algorithms were significantly improved in \cite{BanBBKLW19, DBLP:journals/corr/abs-1807-07156}, where efficient  polynomial time approximation schemes (EPTASes)  were obtained.

Also note that for general (non-binary) matrices a significant amount of work is devoted to $L1$-PCA, where one seeks a low-rank matrix $\bfB$ approximating given matrix $\bfA$ in \emph{entrywise} $\ell_1$ norm, see e.g. \cite{SongWZ17}.
\medskip

While our main motivation stems from low-rank matrix approximation problems, \LRAGF extends   \textsc{Closest String},  very well-studied problem about strings. 
Given a set of binary strings $S = \{s_1 , s_2 , \dots , s_n \}$, each of length $m$, the \textsc{Closest String} problem is to find the smallest $d$ and a string $s$ of length $m$ which is within Hamming distance $d$ to each $s_i \in S$. 

A long history of algorithmic improvements for  \textsc{Closest String} was concluded by 
 the PTAS of running time  $n^{\OO(\epsilon^{-5})}$   by 
 Li, Ma, and Wang
\cite{LiMW02}, which running time was later improved to $n^{\OO(\epsilon^{-2})}$~\cite{MaS09}. Let us note that   \textsc{Closest String} can be seen as a special case of  \LRAGF  for $r=1$. Indeed,
   \textsc{Closest String}  is exactly the  variant of \LRAGF, where columns of $\bfA$ are strings of $S$ and  approximating matrix $\bfB$ is required to have all columns equal.  Note that in  a binary matrix $\bfB$ of rank $1$ all non-zero columns are equal. However,  it is easy to construct an equivalent instance of \textsc{Closest String} by attaching to each string of $S$ a string   $1^{m+1}$, such that the solution  to \LRAGF for $r=1$ does not have zero columns.

  Cygan et al.  \cite{cygan_et_al:LIPIcs:2016:6023} proved that the existence of an EPTAS for  \textsc{Closest String},
  that is  $(1+\varepsilon)$-approximation in time $n^{\OO(1)} \cdot f(\varepsilon)$, for any computable  function $f$, 
   is   unlikely, as it would imply that FPT$=$W[1], a highly unexpected collapse in the hierarchy of parameterized complexity classes. They also showed that the existence of a PTAS for  \textsc{Closest String} with running time $f(\varepsilon) n^{o(1/\varepsilon)}$, for any computable function $f$, would contradict the Exponential Time Hypothesis. The result of Cygan et al.  implies that \LRAGF  also does not admit EPTAS (unless FPT$=$W[1]) already for   $r=1$.
   
A  generalization of   \textsc{Closest String},
\textsc{$k$-closest strings} is also known to admit a PTAS
 \cite{JiaoXL04,GasieniecJL04}.  
 This problem corresponds to the  variant of \LRAGF, where    approximating matrix $\bfB$ is required to have at most $k$ different columns. However, it is not clear how solution to this special case can be adopted to solve 
 \LRAGF.

\paragraph*{Our approach.} 
The usual toolbox of techniques to  handle  NP-hard variants of low-rank matrix approximation problems  like  sketching~\cite{Woodr14}, sampling,  and dimension reduction~\cite{BlumHK17} is based on   randomized linear algebra.  It is very unclear whether any 
 of these techniques can be used to solve  even the simplest case  of  \LRAGF with $r=1$. For example for sampling, the presence of just one outlier  outside of a sample, makes all information
we can deduce from the sample about  the column sum norm of the  matrix,
%
   completely useless.  This is exactly the reason why approximation  algorithms for   \textsc{Closest String}
do not rely on such techniques. 
On the other hand, randomized  dimension reduction appears to be very helpful as a ``preprocessing'' procedure whose application allows us to solve \LRAGF by  applying   linear programming techniques similar to the ones  developed for the   \textsc{Closest String}.
From a very general perspective, our algorithm consists of three steps. While each of these steps is based on the previous works, the way to combine these steps, as well as the correctness proof, is a non-trivial task. We start with a high-level description of the steps and then provide more technical   explanations.

\medskip\noindent\textbf{Step~1.} In order to solve \LRAGF, we encode it as the  \rcenter problem. 
This initial step is almost identical to the encoding used in \cite{DBLP:journals/corr/abs-1807-07156} for \bmfgfr. 
Informally, \rcenter is defined as follows. For a given set of binary vectors $X$, a positive integer $k$,  and a set of constraints, we want to find $k$ binary vectors $C=(\bfc_1,\dots, \bfc_k)$ satisfying the constraints and minimizing  $\max_{\bfx \in X} \hdist(\bfx,C)$, where $\hdist(\bfx,C)$ is the  Hamming distance   between 
$\bfx$ and  the closest vector from $C$.  
For example, when $k=1$ and there are no constraints, then this is just 
the \textsc{Closest String} problem over binary alphabet. 

In the technical description below we give a formal definition of this encoding and in Section~\ref{sec:applications} we prove that \LRAGF is a special case of  \rcenter. Now on, we are working with  \rcenter.

\medskip\noindent\textbf{Step~2.}  We give an approximate Turing reduction which allows to find a partition of vector set $X$ into clusters $X_1, \dots, X_k$ such that if we find a tuple of vectors 
$C=(\bfc_1,\dots, \bfc_k)$ satisfying the constraints and minimizing 
$\max_{1\leq i\leq k,\bfx\in X_i} \hdist(\bfx,\{\bfc_i\})$, then the same tuple $C$ will be a good approximation to \rcenter. In order to obtain such a partition, we use the dimension reduction technique of Ostrovsky and Rabani~\cite{OstrovskyR02}.  While this provides us with important structural information,  we are not done  yet. Even with a given partition, the task of finding  the corresponding tuple of ``closest strings'' $C$ satisfying the constraints, is non-trivial.

\medskip\noindent\textbf{Step~3.} In order to find the centers, we implement the  approach used by    Li, Ma, and Wang
in \cite{LiMW02} to solve \textsc{Closest String}. By brute-forcing, it is possible to reduce the solution of the problem to   special instances, which loosely speaking, have a large optimum.  
Moreover,  \rcenter\  has an Integer Programming (IP)  formulation. 
  Similar to \cite{LiMW02}, for the reduced instance of \rcenter (which has a  ``large optimum'') 
it is possible to prove that the randomized rounding of the corresponding Linear Program (LP) relaxation of this IP,  provides a good approximation.

\medskip
Now  we give a more technical description of the algorithm.

\smallskip\noindent\textbf{Step~1. \rcenter.} 
To define  \rcenter, we need to define some notations. 
A $k$-ary relation $R$ is a set of binary $k$-tuples with elements from $\{0,1\}$. A $k$-tuple $t=(t_1,\dots, t_k)$ \emph{satisfies} $R$, we write $t\in R$, if $t$ is equal to one of the $k$-tuples in $R$.  
\begin{definition}[Vectors satisfying $\cR$]
    Let $\cR=(R_1, \dots, R_m)$ be a tuple of $k$-ary relations. We say that a tuple $C=(\bfc_1, \bfc_2, \dots, \bfc_k)$ of binary $m$-dimensional vectors  \emph{satisfies $\cR$} and write $<C,\cR>$, if 
 $(\bfc_1[i],\ldots,\bfc_k[i])\in R_i$ for all $i\in \{1,\ldots,m\}$.
\end{definition}

For example, for $m=2$, $k=3$, $R_1=\{(0,0,1), (1,0,0)\}$, and  $R_2=\{(1,1,1), (1,0,1), (0,0,1)\}$, the tuple of vectors 
\[
\bfc_1=\left(
\begin{array}{c}
0\\
1\\
\end{array}
\right) , \, 
\bfc_2=\left(
\begin{array}{c}
0\\
0\\
\end{array}
\right) , \,
\bfc_3=\left(
\begin{array}{c}
1\\
1\\
\end{array}
\right)   \]
satisfies $\cR=(R_1,  R_2)$ because  $(\bfc_1[1],\bfc_2[1],\bfc_3[1])=(0,0,1)\in   R_1$ and $({\bfc}_1[2],\bfc_2[2],\bfc_3[2])=(1,0,1)\in   R_2$.

Let us recall that the \emph{Hamming distance} between two vectors $\bfx, \bfy\in\{0,1\}^m$, where $\bfx=(x_1,\ldots,x_m)^\intercal$ and $\bfy=(y_1,\ldots,y_m)^\intercal$, is $\hdist(\bfx,\bfy)=\sum_{i=1}^m |x_i-y_i|$ or, in   words, the number of positions $i\in\{1,\ldots,m\}$ where $x_i$ and $y_i$ differ. 
Recall that 
for a set of vectors $C\subseteq \{0,1\}^m$ and a vector $\bfx\in \{0,1\}^m$,  
$\hdist(\bfx,C)=\min_{\bfc\in C}\hdist(\bfx,\bfc)$.   For sets $X,C\subset \{0,1\}^m$, we 
define  \(\Hdist(X,C)=\max_{\bfx\in X}\hdist(\bfx,C).\)

Now we define \rcenter formally. 

\defproblem{\rcenter}{A set $X\subseteq \{0,1\}^m$ of $n$ vectors, a positive integer $k$, and a tuple of $k$-ary relations
$\cR=(R_1, \dots, R_m)$. }{Among all  tuples $C=(\bfc_1,\ldots,\bfc_k)$ of vectors from $\{0,1\}^m$ satisfying $\cR$, find a tuple $C$ minimizing  $\Hdist(X,C)$.}

As in the case of \bmfgfr in \cite{DBLP:journals/corr/abs-1807-07156}, we prove that \LRAGF is a special case of 
\rcenter, where $k=2^r$. For completeness, this proof and other applications of \rcenter are given in Section~\ref{sec:applications}. Thus, to prove Theorem~\ref{thm:norm1PTAS}, it is enough to design a PTAS for \rcenter.

\begin{restatable}{theorem}{thmpartitionprob}
\label{thm:kclustering}
There is an algorithm for \rcenter\ that given an instance $J=(X,k,\RR)$ and $0<\varepsilon<1$, runs in time 
$m^{\OO(1)}n^{\OO((k/\varepsilon)^4)}$, and outputs 
a $(1+\varepsilon)$-approximate solution 
with probability at least $1-2n^{-2}$. 
\end{restatable}

By the argument above, Theorem~\ref{thm:norm1PTAS} is an immediate corollary of Theorem~\ref{thm:kclustering}. 

\medskip\noindent\textbf{Step~2: Dimension reduction.}  
Let $J=(X,k,{\cal R}=(R_1,\ldots,R_m))$ be an instance of \rcenter  and $C=(\bfc_1,\ldots,\bfc_k)$ be a solution to $J$, that is, a tuple of vectors satisfying ${\cal R}$. Then, the cost of  $C$ is $\Hdist(X,C)$. Given the tuple $C$, there is a natural way we can partition the set of vectors $X$ into $k$ parts $X_1\uplus\cdots\uplus X_k$ such that \[\Hdist(X,C)=\max_{{i\in \{1,\ldots,k\}},\bfx\in X_i}\hdist(\bfx,\bfc_i).\] 
Thus, for each vector $\bfx$ in $X_i$, the closest to $\bfx$ vector from $C$ is  $\bfc_i$. 
We call  such a partition $X_1\uplus\cdots\uplus X_k$ the \emph{clustering of $X$ induced by $C$} and refer to the sets $X_1,\ldots,X_k$ as the  {\em clusters corresponding to   $C$}. 
%
We use  $\opt(J)$  to denote the cost of an optimal solution to $J$. That is, 
\(
\opt(J)=\min \{  \Hdist(X,C) ~| ~ <C,\cR>  \}.
\)
In fact, even if we know the clustering of $X$ induced by a hypothetical optimal solution, finding a good solution is not trivial as the case when $k=1$ is the same as the \textsc{Closest String} problem. 

As mentioned before, our approach is to reduce to a version of  \rcenter, where we know the partition of $X$, and solve the corresponding problem. 
That is, 
we design an  approximation scheme for  the following partitioned  version of the problem.

\defproblem{\rpartitioncenter}{A positive integer $k$, a set $X\subseteq \{0,1\}^m$ of $n$ vectors partitioned into $X_1\uplus\ldots\uplus X_k$,  and a tuple of $k$-ary relations
$\cR=(R_1, \dots, R_m)$. }{Among all tuples $C=(\bfc_1,\ldots,\bfc_k)$ of vectors from $\{0,1\}^m$ satisfying $\cR$, find a tuple $C$ minimizing  
$\max_{i\in \{1,\ldots,k\},\bfx\in X_i} \hdist(\bfx,\bfc_i)$.}

For an instance 
$J'=(k,X=X_1\uplus\ldots X_k,\RR)$ of \rpartitioncenter, we use $\opt(J')$ to denote the cost of an optimal solution to $J'$. 
That is, \[\opt(J')=\min_{C=(\bfc_1,\ldots,\bfc_k) \mbox{ s.t. } <C,\RR>} \left\{  \max_{i\in \{1,\ldots, k\},\bfx\in X_i}\hdist(\bfx,\bfc_i)   \right\}.\]

Clearly, for an instance $J=(X,k,\RR)$ of \rcenter  and a partition of $X$ into $X_1\uplus\ldots \uplus X_k$, any solution to the instance  
$J'=(k,X=X_1\uplus\ldots X_k,\RR)$ of \rpartitioncenter, of cost  $d$, is also a solution to $J$ with cost at most $d$. 
We prove that there is a randomized polynomial time algorithm that given an instance $J=(X,k,\RR)$ of \rcenter  and $0<\epsilon\leq \frac{1}{4}$, outputs a collection ${\cal I}$ of \rpartitioncenter instances $J'=(k,X=X_1\uplus\ldots \uplus X_k,\RR)$ such that  the cost of at least one instance in ${\cal I}$ is at most $(1+4\epsilon)\opt(J)$ with high probability.

\begin{restatable}{lemma}{lempartcreation}
\label{lem:partitioncreation}
There is an algorithm that given an instance $J=(X,k,\RR)$ of \rcenter, $0<\epsilon\leq \frac{1}{4}$, and $\gamma>0$, runs in time $m^2 n^{\OO(k/\epsilon^4)}$, and outputs a collection ${\cal I}$ of $m\cdot n^{\OO(k/\epsilon^4)}$ instances of \rpartitioncenter such that each instance in ${\cal I}$ is of the form $(k,X=X_1\uplus\ldots \uplus X_k,\RR)$, and there exists $J'\in {\cal I}$ such that $\opt(J')\leq (1+4\epsilon)\opt(J)$ with probability at least $1-n^{-\gamma}$. 
\end{restatable}

To prove Lemma~\ref{lem:partitioncreation}, we use the dimension reduction technique of Ostrovsky and Rabani from  \cite{OstrovskyR02}. Loosely speaking, this  technique provides a linear map $\psi$ with the following properties.  For any $\bfy\in \{0,1\}^m$, $\psi(\bfy)$ is a 0-1 vector of length $\OO(\log n/\epsilon^4)$, and for any set $Y$ of $n+k$ vectors,  Hamming distances between any pair of vectors in $\psi(Y)$ are relatively preserved   
with high probability.  So we assume that $\psi$ is ``a good map'' for the set of vectors $X\cup C$, where $C=(\bfc_1,\ldots,\bfc_k)$ is a hypothetical optimal solution to $J$.  Then, we guess the potential tuples of vectors $(\phi(\bfc_1),\ldots,\phi(\bfc_k))$  for the hypothetical optimal solution $C=(\bfc_1,\ldots,\bfc_k)$, and use these choices for $(\phi(\bfc_1),\ldots,\phi(\bfc_k))$ to construct partitions of $X$, and thereby construct instances in ${\cal I}$.  Lemma~\ref{lem:partitioncreation} is proved in Section~\ref{sec:dimred}. 


\medskip\noindent\textbf{Step~3:  LP relaxation.} 
Because of Lemma~\ref{lem:partitioncreation}, to prove Theorem~\ref{thm:kclustering}, it is enough to design a PTAS for 
\rpartitioncenter. So we prove the following lemma. 

\begin{restatable}{lemma}{lempartitionprob}
\label{lem:partkclustering}
There is an algorithm for \rpartitioncenter\ that given an instance $J=(k,X=X^i_1\uplus\ldots\uplus X^i_k,\RR)$ and $0<\epsilon<1/2$, runs in time 
$m^{\OO(1)}n^{\OO((k/\epsilon)^4)}$, and outputs a solution of cost at most $(1+\epsilon)\opt(J)$ with probability at least $1-n^{-2}$. 
\end{restatable}

Towards the proof of Lemma~\ref{lem:partkclustering}, we  encode \rpartitioncenter\ using an  Integer programming (IP) formulation (see \eqref{eqn:obtprob} in Section~\ref{sec:lp}).  We show that the randomized rounding using the solution of the linear programming relaxation of this IP provides a good approximation if the optimum value is large.  Here we follow the approach similar to the  one used by    Li, Ma, and Wang
in \cite{LiMW02} to solve \textsc{Closest String}. We prove that there exist $Y_1\subseteq X_1,\ldots,Y_k\subseteq X_k$, each of size $r=1+\frac{4}{\epsilon}$, with the following property. Let $Q$ be the set of positions in $\{1,\ldots,m\}$ such that for each $i\in \{1,\ldots,k\}$ and $j\in Q$, all the vectors in $Y_i$ agree at the position $j$, and for each $j \in Q$, $(\bfy_1[j],\ldots,\bfy_k[j])\in R_j$, where 
$\bfy_i\in Y_i$ for all $i\in \{1,\ldots,k\}$. Then, for any  solution of $J$ such that for each $j\in Q$ the entries at the position $j$ coincide with $(\bfy_1[j],\ldots,\bfy_k[j])$, the cost of this solution restricted to $Q$ deviates from the cost of an optimal solution restricted to $Q$ by at most $\frac{1}{r-1}\opt(J)$. Moreover, the subproblem of $J$ restricted to $\{1,\ldots,m\}\setminus Q$ has large optimum value and we could use linear programming to solve the subproblem. Lemma~\ref{lem:partkclustering} is proved in Section~\ref{sec:lp}. 

\paragraph*{Putting together.} 
Next we explain how to prove Theorem~\ref{thm:kclustering} using Lemmata~\ref{lem:partitioncreation} and \ref{lem:partkclustering}.  Let $J=(X,k,\RR)$ be the input instance of \rcenter\ and 
$0<\varepsilon<1$ be the given error parameter. Let $\beta=\frac{\varepsilon}{8}$. Since $\varepsilon<1$, $\beta<\frac{1}{4}$. Now, we apply Lemma~\ref{lem:partitioncreation}  on $J$, $\beta$, and $\gamma=2$. As a result, we get a collection ${\cal I}$ of instances of \rpartitioncenter\ such that each instance in ${\cal I}$ is of the form $(k,X=X_1\uplus\ldots \uplus X_k,\RR)$, and there exists $J'\in {\cal I}$ such that $\opt(J')\leq (1+4\beta)\opt(J)$ with probability at least $1-n^{-2}$. From now on, we assume that this event happened.  Next, for each instance $\widehat{J}\in {\cal I}$, we apply Lemma~\ref{lem:partkclustering} with the error parameter $\beta$, and output the best solution among the solutions produced. Let $J'\in {\cal I}$ be the instance 
such that $\opt(J')\leq (1+4\beta)\opt(J)\leq (1+\frac{\varepsilon}{2})\opt(J)$.  Any solution to $\widehat{J}\in {\cal I}$ of cost $d$,  is also a solution to $J$ of cost at most $d$. Therefore, because of Lemmas~\ref{lem:partitioncreation} and \ref{lem:partkclustering}, our algorithm outputs a solution of $J$ with cost at most $(1+\beta)\opt(J')=(1+\frac{\varepsilon}{8})(1+\frac{\varepsilon}{2})\opt(J)\leq (1+\varepsilon)\opt(J)$ with probability at least $1-2n^{-2}$, since both Lemmas~\ref{lem:partitioncreation} and \ref{lem:partkclustering} have the success probability of at least $1 - n^{-2}$. The running time of the algorithm follows from Lemmata~\ref{lem:partitioncreation} and \ref{lem:partkclustering}.

As Theorem~\ref{thm:kclustering} is already proved using  Lemmas~\ref{lem:partitioncreation} and \ref{lem:partkclustering}, the rest of the paper is devoted to the proofs of Lemmata~\ref{lem:partitioncreation} and \ref{lem:partkclustering}, and to the examples of the expressive power of \rcenter, including \LRAGF. In Sections~\ref{sec:dimred} and \ref{sec:lp}, we prove Lemmata~\ref{lem:partitioncreation} and \ref{lem:partkclustering}, respectively. In Section~\ref{sec:applications}, we give applications of Theorem~\ref{thm:kclustering}.

\section{Proof of Lemma~\ref{lem:partitioncreation}}
\label{sec:dimred}

In this section we prove Lemma~\ref{lem:partitioncreation}. The main idea is to map the given instance to a low-dimensional space while approximately preserving distances, then try all possible tuples of centers in the low-dimensional space, and construct an instance of \rpartitioncenter by taking the optimal partition of the images with respect to a fixed tuple of centers back to the original vectors.

To implement the mapping, we employ the notion of $(\delta,\ell,h)$-distorted maps, introduced by Ostrovsky and Rabani \cite{OstrovskyR02}. Intuitively, a $(\delta, \ell, h)$-distorted map approximately preserves distances between $\ell$ and $h$, does not shrink distances larger than $h$ too much, and does not expand  distances smaller than $\ell$ too much. In what follows we make the definitions formal.

A metric space is a pair $(P, d)$ where $P$ is a set (whose elements are called
points), and $d$ is a distance function $d : P \times P \rightarrow {\mathbb R}$ (called a metric), such that
for every $p_1,p_2,p_3 \in P$ the following conditions hold:  
$(i)$ $d(p_1,p_2 )\geq 0$,
$(ii)$ $d(p_1,p_2)=d(p_2,p_1)$,
$(iii)$ $d(p_1,p_2) =0$ if and only if  $p_1 = p_2$, and 
$(iv)$ $d(p_1,p_2 )+d(p_2, p_3) \geq d(p_1,p_3)$. 
Condition $(iv)$ is called the triangle inequality. 
The pair $(\{0,1\}^m,\hdist)$, binary vectors of lentgh $m$ and the Hamming distance, is a metric space.


\begin{definition}[\cite{OstrovskyR02}]
\label{def:distortion}
Let $(P,d)$ and $(P',d')$ be two metric spaces. Let $X,Y \subseteq P$. 
Let $\delta,\ell,h$ be such that $\delta > 0$ and $h>\ell\geq 0$. A mapping $\psi : P \rightarrow P'$
is $(\delta,\ell,h)$-distorted on $(X,Y)$  if and only if  there exists $\alpha > 0$ such that for every 
$x \in X$ and $y \in Y$, the following conditions hold. 
\begin{enumerate}
 \item If $d(x,y) <\ell$, then $d(\psi(x),\psi(y)) < (1 + \delta)\alpha \ell$.
 \item If $d(x,y) >h$, then $d(\psi(x),\psi(y)) > (1-\delta)\alpha h$.
 \item  If $\ell \leq d(x, y) \leq h$, then $(1-\delta)\alpha d(x,y) \leq  d(\psi(x), \psi(y)) \leq (1 + \delta)\alpha d(x, y)$.
 \end{enumerate}
 If $X=Y$, then we say that $\psi$ is $(\delta,\ell,h)$-distorted on $X$.
\end{definition}


For any $r,r'\in {\mathbb N}$ and $\varepsilon>0$, ${\cal A}_{r,r'}(\varepsilon)$ denotes a distribution over $r'\times r$ binary matrices $M\in \{0,1\}^{r'\times r}$, where entries are independent, identically distributed, random $0/1$ variables with $\Pr[1]=\varepsilon$.   

\begin{proposition}[\cite{OstrovskyR02}]
\label{prop:distortion0}
Let $m,\ell\in {\mathbb N}$, and let $X\subseteq \{0,1\}^m$ be a set of $n$ vectors.  
For every $0 < \epsilon \leq 1/2$, there exists a mapping $\phi : X \rightarrow \{0,1\}^{m'}$, where $m' = \OO(\log n/\epsilon^4)$, which is $(\epsilon, \ell/4, \ell/2\epsilon)$-distorted on $X$ (with respect to the Hamming distance in both spaces). 
More precisely, for every $\gamma > 0$ there exists $\lambda > 0$, such that, setting $m' = \lambda \log n/\epsilon^4$, the linear map $\bfx  \mapsto A\bfx$, where $A$ is a random matrix drawn from ${\cal A}_{m,m'}(\epsilon^2/\ell)$, is  $(\epsilon, \ell/4, \ell/2\epsilon)$-distorted on $X$ with probability at least $1-n^{-\gamma}$. 
\end{proposition}


Now we are ready to prove Lemma~\ref{lem:partitioncreation}.  We restate it for convenience.

\lempartcreation*

\begin{proof}
    Without loss of generality, we may assume $\opt(J) > 0$. If $\opt(J) = 0$, there are at most $k$ distinct vectors in $X$, and we trivially construct a single instance of \rpartitioncenter by grouping equal vectors together.

Let $n=\vert X\vert$ and $n'=n+k$. 
Let $\lambda=\lambda(\gamma)$ be the constant mentioned in Proposition~\ref{prop:distortion0}, and $m' = \lambda \log n'/\epsilon^4$. Then, for each $\ell\in [m]$\footnote{For an integer $n\in {\mathbb N}$, we use $[n]$ as a shorthand for $\{1,\ldots,n\}$.}, we construct 
the collection ${\cal I}_{\ell}$ of $n^{\OO(k/\epsilon^4)}$ \rpartitioncenter instances
as follows.
\begin{itemize}
    \item Start with ${\cal I}_{\ell}:=\emptyset$.
\item Randomly choose a matrix $A^{\ell}$ from the distribution  ${\cal A}_{m,m'}(\epsilon^2/\ell)$.
\item For each choice of $k$ vectors $\bfc_1',\ldots,\bfc_k'\in \{0,1\}^{m'}$, construct a partition $X_1\uplus\ldots\uplus X_k$ of $X$ 
such that for each $\bfx\in X_i$, $\bfc_i'$ is one of the closest vectors to $A^{\ell}\bfx$ among $C'=\{\bfc_1',\ldots,\bfc_k'\}$.  Then, add $(k,X=X^i_1\uplus\ldots X^i_k,\RR)$ to ${\cal I}_{\ell}$. 
\end{itemize}

Finally, our algorithm  outputs ${\cal I}=\bigcup_{\ell\in [m]} {\cal I}_{\ell}$ as the required collection of \rpartitioncenter instances. 
Notice that for any $\ell\in [m]$, $\vert {\cal I}_{\ell}\vert=2^{m'k}=n^{\OO(k/\epsilon^4)}$. This implies that the cardinality of $\cal I$ is upper bounded by $m\cdot n^{\OO(k/\epsilon^4)}$, and the construction of ${\cal I}_{\ell}$ takes time $m\cdot n^{\OO(k/\epsilon^4)}$. Thus, the total running time of the algorithm is $m^2\cdot n^{\OO(k/\epsilon^4)}$.

Next, we prove the correctness of the algorithm. Let $\ell=\opt(J)$ and $C=(\bfc_1,\ldots,\bfc_k)$ be an optimum solution of $J$. 
Let $Y_1,\ldots,Y_k$ be the clusters corresponding to $C$. Consider the step in the algorithm where we constructed ${\cal I}_{\ell}$. 
By Proposition~\ref{prop:distortion0}, the map $\psi \colon \bfx  \mapsto A^{\ell}\bfx$ is $(\epsilon, \ell/4, \ell/2\epsilon)$-distorted on $X\cup C$ with probability at least $1-n^{-\gamma}$. 
In the rest of the proof, we assume that this event happened. 
Let $\bfc_1'=A^{\ell}\bfc_1,\ldots,\bfc_k'=A^{\ell}\bfc_k$. 
Consider the \rpartitioncenter instance constructed for the choice of vectors $\bfc_1',\ldots,\bfc'_k$.  That is, let $X_1,\ldots,X_k$ be the partition of $X$ such that for each $\bfx\in X_i$,  $\bfc_i'$ is one of the closest vector to $A^{\ell}\bfx$ from $C'=\{\bfc_1',\ldots,\bfc_k'\}$. 
Let $J'$ be the instance $(k,X=X^i_1\uplus\ldots X^i_k,\RR)$ of \rpartitioncenter. 

Now, we claim that $C$ is a solution to $J'$ with cost at most $(1+4\epsilon)\ell=(1+4\epsilon)\opt(J)$. Since $C$ satisfies $\RR$, $C$ is a solution of $J'$. 
To prove $\opt(J')\leq (1+4\epsilon)\ell$, it is enough to prove that for each $i\in [k]$ and $\bfx\in X_i$, $\hdist(\bfx,\bfc_i)\leq (1+4\epsilon)\ell$. Fix an index $i\in [k]$ and $\bfx\in X_i$. Suppose $\bfx\in Y_i$. 
Since $C$ is an optimum solution of $J$ with corresponding clusters $Y_1,\ldots Y_k$, we have that $\hdist(\bfy,\bfc_i)\leq \ell$ for all $\bfy\in Y_i\cap X_i$. Thus,  $\hdist(\bfx,\bfc_i)\leq \ell$. So, now consider the case $\bfx\in Y_j$ for some $j\neq i$. Notice that if $\hdist(\bfx,\bfc_i)\leq \ell$, then we are done. We have the following two subcases. 

\paragraph*{Case 1: $\hdist(\bfx,\bfc_i)\leq \frac{\ell}{2\epsilon}$.} 

We know that the map $\psi \colon \bfx  \mapsto A^{\ell}\bfx$ is $(\epsilon, \ell/4, \ell/2\epsilon)$-distorted on $X\cup C$,  and let $\alpha>0$ be the number such that conditions of Definition~\ref{def:distortion} hold. 
Since $\bfx\in X_i$, we have that $(a)$ $\hdist(\psi(\bfx),\psi(\bfc_i))\leq \hdist(\psi(\bfx),\psi(\bfc_j))$.  Since $\hdist(\bfx,\bfc_j)\leq \ell$ (because $\bfx\in Y_j$) and $\psi$ is $(\epsilon, \ell/4, \ell/2\epsilon)$-distorted on  $X\cup C$, we have that 
$(b)$ $\hdist(\psi(\bfx),\psi(\bfc_j))\leq (1+\epsilon)\alpha \ell$. Since $\ell<\hdist(\bfx,\bfc_i)\leq \frac{\ell}{2\epsilon}$, and $\psi$ is $(\epsilon, \ell/4, \ell/2\epsilon)$-distorted on  $X\cup C$, we have that $(c)$ $(1-\epsilon)\alpha \hdist(\bfx,\bfc_i)\leq \hdist(\psi(\bfx),\psi(\bfc_i))$. The statements $(a)$, $(b)$, and $(c)$ imply that 
\[
\hdist(\bfx,\bfc_i) \leq \frac{1+\epsilon}{1-\epsilon}\ell \leq (1+4\epsilon)\ell, 
\]
where the last inequality holds since $\epsilon \leq 1/4$. 
\paragraph*{Case 2: $\hdist(\bfx,\bfc_i)> \frac{\ell}{2\epsilon}$.} We prove that this case is impossible by showing a contradiction. Since $\epsilon \leq 1/4$, in this case, we have that $\hdist(\bfx,\bfc_i)>2\ell$. Since $\psi$ is $(\epsilon, \ell/4, \ell/2\epsilon)$-distorted on  $X\cup C$, $\hdist(\bfx,\bfc_i)>2\ell$, and 
$\hdist(\bfx,\bfc_j)\leq \ell$, we have that 
\begin{eqnarray*}
    (1-\epsilon) \alpha \cdot 2\ell \leq \hdist(\psi(\bfx),\psi(\bfc_i))\leq \hdist(\psi(\bfx),\psi(\bfc_j))\leq (1+\epsilon) \alpha \cdot \ell.
\end{eqnarray*}
Then $2(1-\epsilon)\leq (1+\epsilon)$ and thus $\epsilon\geq 1/3$, which contradicts the assumption that $\epsilon \le 1/4$.  

This completes the proof of the lemma. 
\end{proof}


\section{Proof of Lemma~\ref{lem:partkclustering}}
\label{sec:lp}

For a set of positions $P \subset [m]$, let us define the Hamming distance restricted to $P$ by
\[\hdist^P(\bfx, \bfy) = \sum_{i \in P} |x_i - y_i|.\]


We use the following lemma  in our proof.  


\begin{lemma}
\label{lemma:ragree}
 Let $Y = \{\bfy_1, \cdots, \bfy_l\} \subset \{0, 1\}^m$ be a set of vectors and $\bfc^*\in \{0,1\}^m$ be a vector. 
Let $d^*=\Hdist(Y,\{\bfc^*\})=\max_{\bfy\in Y} \hdist(\bfy,\bfc^*)$. For any $r \in \mathbb{N}$, $r > 2$, there exist indices $i_1$, \ldots, $i_r$ such that for any $\bfx \in Y$
    \[\hdist^P(\bfx, \bfy_{i_1}) - \hdist^P(\bfx, \bfc^*) \le \frac{1}{r - 1}d^*,\]
    where $P$ is any subset of $Q_{i_1, \ldots, i_r}$ and $Q_{i_1, \ldots, i_r}$ is the set of positions where all of $\bfy_{i_1}, \ldots, \bfy_{i_r}$ coincide (i.e., $Q_{i_1, \ldots, i_r}=\{j\in[m] \colon \bfy_{i_1}[j]=\bfy_{i_2}[j]=\ldots=\bfy_{i_r}[j]\}$).
\end{lemma}

\begin{proof}
%
%

For a vector $\bfx=\bfy_{\ell'}\in Y$ and $P\subseteq Q_{i_1, \ldots, i_r}$, let 
\begin{eqnarray*}
J_P(\ell') &=& \left\{j \in P \colon \bfy_{i_1}[j] \ne \bfx[j] \text{ and } \bfy_{i_1}[j] \ne \bfc^*[j]\right\}, \mbox{ and } \\
J(\ell') &=& \left\{j \in Q_{i_1, \ldots, i_r} \colon \bfy_{i_1}[j] \ne \bfx[j] \text{ and } \bfy_{i_1}[j] \ne \bfc^*[j]\right\}.
\end{eqnarray*}
To prove the lemma it is enough to prove that  $\vert J_P(\ell')\vert \leq \frac{1}{r - 1}d^*$. Also, since $J_P(\ell')\subseteq J(\ell')$, 
to prove the lemma, it is enough to prove that $\vert J(\ell')\vert \leq \frac{1}{r - 1}d^*$. 
Recall that for any  $s\in [\ell]$ and $1\leq i_1, \ldots, i_s\leq \ell$,   $Q_{i_1, \ldots, i_s}$ is the set of positions where all of $\bfy_{i_1}, \ldots, \bfy_{i_s}$ coincide. For any $2\leq s \le r + 1$ and $1\leq i_1, \ldots, i_s\leq \ell$, let $p_{i_1, \ldots, i_s}$ be the number of mismatches between 
$\bfy_{i_1}$ and $\bfc^*$ at the positions in $Q_{i_1, \ldots, i_s}$. Let 
\[
\rho_s=\min_{1\leq i_1, \ldots, i_s\leq n} \frac{p_{i_1, \ldots, i_s}}{d^*}.
\]
Notice that for any $2\leq s \le r + 1$, $\rho_s\leq 1$. 
\begin{claim}[Claim~2.2~\cite{LiMW02}]\footnote{We remark that  Claim~2.2 in \cite{LiMW02} is stated for a vector $\bfc$ such that $d^*=\Hdist(Y,\{\bfc\})=\min_{\bfc'}\Hdist(Y,\{\bfc'\})$. But the steps of the same proof work in our case as well.}
For any $s$ such that $2\leq s \leq r$, there are indices  $1\leq  i_1,i_2,\ldots,i_r\leq \ell$ such that for any $\bfx=\bfy_{\ell'}\in Y$, $\vert J(\ell')\vert \leq (\rho_s-\rho_{s+1})d^*$. 
\end{claim}
\begin{proof}
Consider indices $1\leq i_1,\ldots,i_s\leq \ell$ such that $p_{i_1, \ldots, i_s}=\rho_s \cdot d^*$. Next arbitrarily pick $r-s$ indices  $i_{s+1},i_{s+2},\ldots,i_r$ from $[\ell]\setminus \{i_1,\ldots,i_s\}$. Next we prove that $i_1,i_2,\ldots,i_r$ are the required set of indices. 
Towards that, fix $\bfx=\bfy_{\ell'}\in Y$, 
\begin{eqnarray}
J(\ell') &=& \vert \left\{j \in Q_{i_1, \ldots, i_r} \colon \bfy_{i_1}[j] \ne \bfx[j] \text{ and } \bfy_{i_1}[j] \ne \bfc^*[j]\right\}\vert \nonumber\\
 &\leq& \vert \left\{j \in Q_{i_1, \ldots, i_s} \colon \bfy_{i_1}[j] \ne \bfx[j] \text{ and } \bfy_{i_1}[j] \ne \bfc^*[j]\right\}\vert  \qquad\quad
(\mbox{Because } Q_{i_1, \ldots, i_r}\subseteq Q_{i_1, \ldots, i_s}) \nonumber\\
&=& \vert \left\{j \in Q_{i_1, \ldots, i_s} \colon \bfy_{i_1}[j] \ne \bfc^*[j]\right\}\setminus \left\{j \in Q_{i_1, \ldots, i_s} \colon \bfy_{i_1}[j] = \bfx[j] \text{ and } \bfy_{i_1}[j] \ne \bfc^*[j]\right\}\vert \nonumber\\
&=& \vert \left\{j \in Q_{i_1, \ldots, i_s} \colon \bfy_{i_1}[j] \ne \bfc^*[j]\right\}\setminus \left\{j \in Q_{i_1, \ldots, i_s,{\ell'}} \colon   \bfy_{i_1}[j] \ne \bfc^*[j]\right\}\vert \quad (\mbox{Because }\bfx=\bfy_{\ell'}) \nonumber\\
&=& \vert \left\{j \in Q_{i_1, \ldots, i_s} \colon \bfy_{i_1}[j] \ne \bfc^*[j]\right\}\vert - \vert \left\{j \in Q_{i_1, \ldots, i_s,{\ell'}} \colon   \bfy_{i_1}[j] \ne \bfc^*[j]\right\}\vert  \label{eqoneof}\\
&=& p_{i_1, \ldots, i_s}-p_{i_1, \ldots, i_s,\ell'} \qquad\qquad\qquad\qquad\qquad\qquad\qquad\qquad\qquad\qquad\qquad (\mbox{By definition})\nonumber\\
&\le&(\rho_s-\rho_{s+1})d^* \label{ineqlityother}
\end{eqnarray}
The equality \eqref{eqoneof} holds since $Q_{i_1, \ldots, i_s}\supseteq Q_{i_1, \ldots, i_s,\ell'}$. The inequality 
\eqref{ineqlityother} holds because $p_{i_1, \ldots, i_s}=\rho_s \cdot d^*$ by the choice of $i_1$, \ldots, $i_s$, and $\rho_{s + 1} d^* \le p_{i_1, \ldots, i_s, \ell'}$ by definition.
This completes the proof of the claim. 
\end{proof}

Notice that 
\(
(\rho_2-\rho_{3})+(\rho_3-\rho_{4})+\ldots+(\rho_r-\rho_{r+1})=(\rho_2-\rho_{r+1})\leq \rho_2\leq 1. 
\)
Thus, one of $(\rho_2-\rho_{3}),(\rho_3-\rho_{4}),\ldots,(\rho_r-\rho_{r+1})$ is at most $1/(r-1)$. This completes the proof of the lemma. 
\end{proof}

Consider the instance $J=(k,X=X^i_1\uplus\ldots X^i_k,\RR)$ of \rpartitioncenter.
Let $C^* = (\bfc^*_1, \cdots, \bfc^*_k)\subset \{0,1\}^m$ be an optimum solution to $J$. 
Let $d_{opt}=\opt(J)=\max_{i \in [k], \bfx \in X_i} \hdist(\bfx, \bfc^*_i)$. 
%
For each $i \in [k]$ and $r\geq 2$, by Lemma~\ref{lemma:ragree}, there exist $r$ elements $\bfx^{(1)}_i$, \ldots, $\bfx^{(r)}_i$ of $X_i$ such that for any $\bfx\in X_i$, 
\begin{equation}
\label{eqn:partialxis}
\hdist^P(\bfx, \bfx^{(1)}_{i}) - \hdist^P(\bfx, \bfc_i^*) \le \frac{1}{r - 1}d_{opt},
\end{equation}
where $P$ is any subset of $Q_i$, and $Q_i$ is the set of coordinates on which $\bfx^{(1)}_i$, \ldots, $\bfx^{(r)}_i$ agree. 
%
Let us denote as $Q$ the intersection of all $Q_i$ from which the positions not satisfying $\mathcal{R}$ are removed. That is, 
\[Q = \left\{j \in \bigcap_{i \in [k]} Q_i  \colon (\bfx^{(1)}_1[j],\bfx^{(1)}_2[j],\ldots,\bfx^{(1)}_k[j]) \in R_j\right\}.\]

Because of \eqref{eqn:partialxis}, there is an approximate solution where the coordinates $j\in Q$ are identified using $\bfx^{(1)}_1,\ldots,\bfx^{(1)}_k$. 
Let $\overline{Q}=[m] \setminus Q$. Now the idea is to solve the problem restricted to $\overline{Q}$ separately, and then complement the solution on $Q$ by the values of $\bfx^{(1)}_i$. We prove that for the `subproblem' restricted on  $\overline{Q}$, the optimum value is {\em large}. Towards that we first prove the following lemma.

\begin{lemma}
\label{lem:firstphase}
Let $J=(k,X=X^i_1\uplus\ldots X^i_k,\RR)$ be an instance of \rpartitioncenter. Let $({\bfc}^*_1,\ldots,{\bfc}^*_k)$ be an optimal solution for $J$, and $r\geq 2$ be an integer. Then, there exist $\{\bfx^{(1)}_1$, \ldots, $\bfx^{(r)}_1\} \subset X_1, \ldots, \{\bfx^{(1)}_k$, \ldots, $\bfx^{(r)}_k\} \subset X_k$  with the following properties. For each $i\in [k]$, let $Q_i$ be the set of coordinates on which $\bfx^{(1)}_i$, \ldots, $\bfx^{(r)}_i$ agree, \(Q = \left\{j \in \bigcap_{i \in [k]} Q_i  \colon (\bfx^{(1)}_1[j],\bfx^{(1)}_2[j],\ldots,\bfx^{(1)}_k[j]) \in R_j\right\}\), and $\overline{Q}=[m] \setminus Q$. 
\begin{itemize}
\item For any  $i\in [k]$ and $\bfx\in X_i$, $\hdist^Q(\bfx, \bfx^{(1)}_{i}) - \hdist^Q(\bfx, \bfc_i^*) \le \frac{1}{r - 1}\opt(J)$, and 
\item  $|\overline{Q}| \le rk \cdot \opt(J)$. 
\end{itemize}
\end{lemma}
\begin{proof}
Fix an integer $i\in [k]$. By substituting $Y=X_i$ and $\bfc^*=\bfc_i^*$ in Lemma~\ref{lemma:ragree}, we get 
$\{\bfx^{(1)}_i$, \ldots, $\bfx^{(r)}_i\} \subset X_i$ such that for any $\bfx\in X_i$, $\hdist^Q(\bfx, \bfx^{(1)}_{i}) - \hdist^Q(\bfx, \bfc_i^*) \le \frac{1}{r - 1}\opt(J)$. That is, we have proved 
the first condition in the lemma.  The second condition is proved in the following claim. 

\begin{claim}
$|\overline{Q}| \le rk \cdot \opt(J)$. 
\end{claim}
\begin{proof}
  We claim that for each position $j \in \overline{Q}$ there exist $i \in [k]$ and $s \in [r]$ such that $\bfx^{(s)}_i[j] \ne {\bfc}^*_i[j]$. There are two kinds of positions in $\overline{Q}$. First, positions, where for some $i \in [k]$ vectors $\bfx^{(1)}_i$, \ldots, $\bfx^{(r)}_i$ do not agree, in this case certainly one of them does not agree with the corresponding position in $\bfc^*_i$. Second, positions $j$ where for any $i \in [k]$, $\bfx^{(1)}_i[j]=\bfx^{(2)}_i[j]  = \cdots = \bfx^{(r)}_i[j]$, but $(\bfx^{(1)}_1[j],\ldots,\bfx^{(1)}_k[j]) \notin R_j$. Then, there exists $i \in [k]$ such that $\bfx^{(1)}_i[j] \ne {\bfc}^*_i[j]$ because $({\bfc}^*_i[j],\ldots,{\bfc}^*_k[j]) \in R_j$.
    
 Now, for any $i \in [k]$ and $s \in [r]$, $\bfx^{(s)}_i$ contributes at most $\opt(J)$ positions to $\overline{Q}$, since $\hdist(\bfx^{(s)}_i, {\bfc}^*_i) \le \opt(J)$. Thus,  in total there are at most $rk \cdot \opt(J)$ positions in $\overline{Q}$.
\end{proof}
This completes the proof of the lemma. 
\end{proof}


As mentioned earlier, we fix the entries of our solution in positions $j$ of $Q$ with values in $\bfx^{(1)}_1[j],\ldots,\bfx^{(1)}_k[j]$.  
Towards finding the entries of our solution in positions  of $\overline{Q}$, we define the following problem and solve it.

\defproblem{\rpartitioncenterstar}{A positive integer $k$, a set $X\subseteq \{0,1\}^m$ of $n$ vectors partitioned into $X_1\uplus\ldots\uplus X_k$,  a tuple of $k$-ary relations
$\cR=(R_1, \dots, R_m)$, and for all $\bfx\in X$, $d_{\bfx}\in {\mathbb N}$ }{Among all tuples $C=(\bfc_1,\ldots,\bfc_k)$ of vectors from $\{0,1\}^m$ satisfying $\cR$, find a tuple $C$ that minimizes the integer $d$ such that  
for all $i\in [k]$ and $\bfx\in X_i$,  $\hdist(\bfx,\bfc_i)\leq d-d_{\bfx}$.}


\begin{restatable}{lemma}{lemlargeopt}
\label{lemma:largeopt}
Let $J'=(k,X=X_1\uplus\ldots X_k,\RR, (d_\bfx)_{\bfx \in X})$ be an instance of \rpartitioncenterstar, $\opt(J')\geq \frac{m}{c}$ for some integer $c$, and $0<\delta<1/c$. Then, there is an algorithm which runs in time $m^{\OO(1)}n^{\OO(c^2k/\delta^2)}$, and outputs a solution $C$ of $J'$, of cost at most $(1 + \delta) \opt(J')$ with probability at least $1-n^{-2}$.
\end{restatable}

Before proving Lemma~\ref{lemma:largeopt}, we explain how all these results puts together to form a proof of Lemma~\ref{lem:partkclustering}.  We restate Lemma~\ref{lem:partkclustering} for the convenience of the reader. 


\lempartitionprob*

\begin{proof}

Let $J=(k,X=X^i_1\uplus\ldots X^i_k,\RR)$ be the input instance of \rpartitioncenter, and $0<\epsilon<\frac{1}{2}$ be the error parameter. 
Let $({\bfc}^*_1,\ldots,{\bfc}^*_k)$ be an optimal solution for $J$. Let $r\geq 2$ be an integer which we fix later.  
First, for each $i \in [k]$ we obtain $r$ vectors $\bfx^{(1)}_i$, \ldots, $\bfx^{(r)}_i \in X_i$ which satisfy the conditions of Lemma~\ref{lem:firstphase}. Their existence is guaranteed by Lemma~\ref{lem:firstphase}, and we guess them in time $n^{\OO(rk)}$ over all $i \in [k]$.
For each $i\in [k]$, let $Q_i$ be the set of coordinates on which $\bfx^{(1)}_i$, \ldots, $\bfx^{(r)}_i$ agree, \(Q = \left\{j \in \bigcap_{i \in [k]} Q_i  \colon (\bfx^{(1)}_1[j],\bfx^{(1)}_2[j],\ldots,\bfx^{(1)}_k[j]) \in R_j\right\}\), and $\overline{Q}=[m] \setminus Q$. Next, we construct a solution $C=(\bfc_1,\ldots,\bfc_k)$ as follows. 
For each $i\in [k]$ and $j\in Q$, we set $\bfc_i[j]=\bfx_i^{(1)}[j]$. 

Towards finding the entries of vectors $\bfc_1,\ldots,\bfc_k$ at the coordinates in $\overline{Q}$, we use  Lemma~\ref{lemma:largeopt}. Let $J'$ be the instance of \rpartitioncenterstar, which is a natural restriction of $J$ to $\overline{Q}$. 
That is,  $J'=(k,X'=X'_1\uplus\ldots X'_k,\RR|_{\overline{Q}}, (d_{\bfx|_{\overline{Q}}})_{\bfx \in X'})$, where 
for each $i\in [k]$, $X_i'=\{\bfx|_{\overline{Q}} \colon \bfx\in X_i\}$ and for each $\bfx \in X_i$, $d_{\bfx|_{\overline{Q}}} = \hdist^Q(\bfx, \bfx^{(1)}_i)$. By the second condition in Lemma~\ref{lem:firstphase}, 
we have that $|\overline{Q}| \le rk \cdot \opt(J)$. 

\begin{claim}
\label{cliamoptcom}
$\opt(J)\leq \opt(J')\leq \left(1+\frac{1}{r-1}\right) \opt(J).$
\end{claim}

\begin{proof}
First,  we prove that $\opt(J)\leq \opt(J')$.  
Towards that we show that we can transform a solution $C' = (\bfc'_1, \cdots, \bfc'_k)$ of $J'$ with the objective value $d$ to a solution $C$ of $J$ with the same objective value. For each $i \in [k]$, consider $\widehat{\bfc}_i$ which is equal to $\bfx^{(1)}_i$ restricted to $Q$, and to $\bfc'_i$ restricted to $\overline{Q}$, and the solution $\widehat{C} = (\widehat{\bfc}_1, \cdots, \widehat{\bfc}_k)$. Clearly, $\widehat{C}$ satisfies $\RR$ since on $\overline{Q}$ it is guaranteed by $C'$ being a solution to $J'$, and on $Q$ by construction of $Q$. 
The objective value of $C$ is 
\begin{eqnarray*}
 \max_{i \in [k],\bfx \in X_i} \hdist(\bfx, \bfc_i) &=& \max_{i \in [k],\bfx \in X_i} \left(\hdist^{\overline{Q}}(\bfx, \bfc_i) + \hdist^Q (\bfx, \bfc_i)\right) \\ 
 &=& \max_{i \in [k],\bfx \in X_i} \left(\hdist(\bfx|_{\overline{Q}}, \bfc'_i) + \hdist^Q (\bfx, \bfx^{(1)}_i)\right)\\
 &=& \max_{i \in [k],\bfx \in X_i} \left(\hdist(\bfx|_{\overline{Q}}, \bfc'_i) + d_{\bfx|_{\overline{Q}}}\right) = d.
\end{eqnarray*}
Thus, $\opt(J)\leq \opt(J')$. 

Next, we prove that $\opt(J')\leq\left(1+\frac{1}{r-1}\right) \opt(J)$. Recall that $({\bfc}^*_1,\ldots,{\bfc}^*_k)$ is an optimal solution for $J$. Then, $({\bfe}^*_1,\ldots,{\bfe}^*_k)$, where each $\bfe^*_i$ is the restriction of $\bfc^*_i$ on $\overline{Q}$, is a solution for $J'$.  For each $i\in [k]$ and $\bfx\in X_i$, 
\begin{eqnarray*}
    \hdist(\bfx|_{\overline{Q}}, \bfe^*_i) + d_{\bfx|_{\overline{Q}}}&=&\hdist^{\overline{Q}}(\bfx, \bfc^*_i) + \hdist^Q (\bfx, \bfx^{(1)}_i)\\
&\leq& \hdist^{\overline{Q}}(\bfx, \bfc^*_i)+\hdist^Q(\bfx, \bfc^*_i) + \frac{1}{r - 1} \opt(J) \qquad\qquad(\mbox{By Lemma~\ref{lem:firstphase}})\\
&\leq& \hdist(\bfx, \bfc^*_i) + \frac{1}{r - 1} \opt(J) \\
&\leq& \left(1+ \frac{1}{r - 1}\right) \opt(J)
\end{eqnarray*}
This completes the proof of the claim. 
\end{proof}

Since $|\overline{Q}| \le rk \cdot \opt(J)$ and by Claim~\ref{cliamoptcom}, we have that $\opt(J')\geq \frac{|\overline{Q}|}{rk}=\frac{|\overline{Q}|}{c}$, where $c=rk$. Let $0<\delta<\frac{1}{c}$ be a number which we fix later.

Now we apply Lemma~\ref{lemma:largeopt} on the input $J'$ and $\delta$, and let $C'=(\bfc'_1$, \ldots, $\bfc'_k)$ be the solution for $J'$ obtained. We know that the cost $d'$ of $\bfc'$ is at most $(1+\delta)\opt(J')$ with probability at least $1-n^{-2}$. 
For the rest of the proof we assume that the cost $d'\leq (1+\delta)\opt(J')$. 
Recall that we have partially computed the entries of the solution $\bfc=(\bfc_1,\ldots,\bfc_k)$ for the instance $J$. That is, for each $j\in Q$ and $i\in [k]$, we have already set the value of $\bfc_i[j]$.   Notice that $C'\subseteq \{0,1\}^{\vert \overline{Q}\vert}$. Since $J'$ is obtained from $J$ by restricting to $\overline{Q}$, there is a natural bijection $f$ from $\overline{Q}$ to $[\vert \overline{Q}\vert]$ such that for 
each $\bfx\in X$ and $j\in {\overline{Q}}$, $\bfx[j]=\bfy[f(j)]$, where $\bfy=\bfx|_{\overline{Q}}$.  Now for each $i\in [k]$ and $j\in \overline{Q}$, we set $\bfc_i[j]=\bfc'_i[f(j)]$.



In Claim~\ref{cliamoptcom}, we have proven that the solution $C$ of $J$ obtained in this way has cost at most $d'$.
By Lemma~\ref{lemma:largeopt}, we know that $d'\leq (1+\delta)\opt(J')$. By Claim~\ref{cliamoptcom}, $\opt(J')\leq (1+\frac{1}{r-1})\opt(J)$. Thus, we have that the cost of the solution $C$ of $J$ is at most $(1+\delta)(1+\frac{1}{r-1})\opt(J)$. Now we fix $r=(1+\frac{4}{\epsilon})$ and $\delta=\frac{\epsilon}{(2\epsilon+8)k}$. Then the cost of $C$ is at most $(1+\epsilon)\opt(J)$. 

\paragraph*{Running time analysis.}
The number of choices for  $\{\bfx^{(1)}_1$, \ldots, $\bfx^{(r)}_1\} \subset X_1, \ldots, \{\bfx^{(1)}_k$, \ldots, $\bfx^{(r)}_k\} \subset X_k$ is at most $n^{\OO(rk)}=n^{\OO(k/\epsilon)}$.  For each such choice, we run the algorithm of Lemma~\ref{lemma:largeopt} which takes time at most $m^{\OO(1)}n^{\OO(c^2k/\delta^2)}=m^{\OO(1)}n^{\OO((k/\epsilon)^4)}$. Thus, the total running time is $m^{\OO(1)}n^{\OO((k/\epsilon)^4)}$. 
\end{proof}

Now the only piece left is the proof of Lemma~\ref{lemma:largeopt}. We use the following tail inequality (a variation of Chernoff bound) in the proof of Lemma~\ref{lemma:largeopt}. 

\begin{proposition}[Lemma 1.2~\cite{LiMW02}]
\label{prop:chernoff}
Let $X_1,\ldots,X_n$ be $n$ independent $0$-$1$ random variables, $X=\sum_{i=1}^n X_i$, and $0<\epsilon\leq 1$.  
Then, \(\Pr[X>E[X]+\epsilon n]\leq e^{-\frac{1}{3}n\epsilon^2}.\)
\end{proposition}

Finally,  we prove Lemma~\ref{lemma:largeopt}. 

\begin{proof}[Proof of Lemma~\ref{lemma:largeopt}]
First, assume that $m < 9 c^2 \log n/\delta^2$. 
If this is the case, we enumerate all possible solutions for $J'$ and output the best solution. The number of solutions is 
at most $2^{k\cdot m}=n^{\OO(c^2 k/\delta^2)}$. Thus, in this case the algorithm is exact and deterministic, and the running time bound holds.
For the rest of the proof we assume that $m \ge 9 c^2 \log n/\delta^2$. 

\rpartitioncenterstar\ can be formulated as a $0$-$1$ optimization problem as explained below. For each $j\in [m]$ and tuple $t\in R_j$, we use a $0$-$1$ variable $y_{j,t}$ to indicate whether the $j^{th}$ entries of a solution form a tuple $t\in R_j$ or not. For any $i\in [k]$, $\bfx\in X_i$,  $j\in [m]$ and $t\in R_j$, denote $\chi_i(\bfx[j],t)=0$ if $\bfx[j]=t[i]$ and 
 $\chi_i(\bfx[j],t)=1$ if $\bfx[j]\neq t[i]$. Now \rpartitioncenterstar\ can be defined as the following $0$-$1$ optimization problem. 
 
 \begin{eqnarray}
 &&\min d \nonumber\\
 &&\mbox{subject to} \nonumber \\
 &&\sum_{t\in R_j} y_{j,t}=1, \qquad\qquad\qquad\qquad\qquad\quad\; \mbox{for all } j\in [m] ; \label{eqn:obtprob}\\
 &&\sum_{j\in [m]} \sum_{t\in R_j} \chi_i(\bfx[j],t)\cdot y_{j,t}\leq d-d_{\bfx}, \qquad \mbox{for all } i\in [k] \mbox{ and } \bfx \in X_i \nonumber\\
&&y_{j,t}\in \{0,1\}, \qquad\qquad\qquad\qquad\qquad\quad\; \mbox{for all } j\in [m] \mbox{ and } t\in R_j \nonumber.
  \end{eqnarray}
Any solution  $y_{j,t}$ ($j\in [m]$ and $t\in R_j$) to \eqref{eqn:obtprob} corresponds to the solution $C=(\bfc_1,\ldots,\bfc_k)$ where for all $j\in [m]$ and $t\in R_j$ such that $y_{j,t}=1$, we have  $(\bfc_1[j],\ldots,\bfc_k[j])=t$. 

Now, we solve the above optimization problem using linear programming relaxation and   obtain a fractional solution 
$y^{\star}_{j,t}$ ($j\in [m]$ and $t\in R_j$) with cost $d'$.  Clearly,  $d'\leq d_{opt}=\opt(J')$. 
Now, for each $j\in [m]$, independently with probability $y^{\star}_{j,t}$, we set $y'_{j,t}=1$ and $y'_{j,t'}=0$, for any $t'\in R_j\setminus \{t\}$. Then  $y'_{j,t}$ ($j\in [m]$ and $t\in R_j$) form a solution to \eqref{eqn:obtprob}. Next we construct the solution $C=(\bfc_1,\ldots,\bfc_k)$ to \rpartitioncenterstar, corresponding to $y'_{j,t}$ ($j\in [m]$ and $t\in R_j$). That is, for all $j\in [m]$ and $t\in R_j$ such that $y_{j,t}=1$, we have  $(\bfc_1[j],\ldots,\bfc_k[j])=t$.

For the running time analysis, notice that solving the linear program and performing the random rounding takes polynomial time in the size of the problem \eqref{eqn:obtprob}. And the size of \eqref{eqn:obtprob} is polynomial in the size of $J'$, so the running time bound is satisfied. It remains to show that the constructed solution has cost at most $(1 + \delta) \opt(J')$ with probability at least $1 - n^{-2}$.

For any $j\in [m]$, the above random rounding procedure ensures that there is exactly one tuple $t\in R_j$ 
such that $y'_{j,t}=1$. This implies that for any $j\in [m]$, $i\in [k]$ and $\bfx\in X_i$, $\sum_{t\in R_j} \chi_i(\bfx[j],t)\cdot y'_{j,t}$ is a $0$-$1$ random variable. Since for each $j\in [m]$ the rounding procedure is independent, we have that for any $i\in [k]$ and $\bfx\in X_i$ the random variables $(\sum_{t\in R_1} \chi_i(\bfx[1],t)\cdot y'_{1,t}), \ldots, (\sum_{t\in R_m} \chi_i(\bfx[m],t)\cdot y'_{j,t})$ are independent. Hence, for any $i\in [k]$ and $\bfx\in X_i$, the Hamming distance between $\bfx$ and $\bfc_i$, $\hdist(\bfx,\bfc_i)=\sum_{j\in [m]}\sum_{t\in R_j} \chi_i(\bfx[j],t)\cdot y'_{j,t}$,  is the  sum of $m$ independent $0$-$1$ random variables.  For each $i\in [k]$ and $\bfx\in X_i$, we upper bound the expected value of   $\hdist(\bfx,\bfc_i)$ as follows. 
\begin{eqnarray*}
E[\hdist(\bfx,\bfc_i)]&=&E\left[\sum_{j\in [m]}\sum_{t\in R_j} \chi_i(\bfx[j],t)\cdot y'_{j,t}\right]\\
&=&\sum_{j\in [m]}\sum_{t\in R_j} \chi_i(\bfx[j],t)\cdot E[y'_{j,t}]\\
&=&\sum_{j\in [m]}\sum_{t\in R_j} \chi_i(\bfx[j],t)\cdot y^{\star}_{j, t}\\
&\leq& d'-d_{\bfx} \qquad\qquad\qquad (\mbox{By the constraints of \eqref{eqn:obtprob}} )
\end{eqnarray*}
Fix $\epsilon=\frac{\delta}{c}$. 
Then, by Proposition~\ref{prop:chernoff}, for all $i\in [k]$, and $\bfx\in X_i$, 
\[
\Pr[\hdist(\bfx,\bfc_i)>d'-d_{\bfx}+\epsilon m] \leq e^{-\frac{1}{3}m\epsilon^2}.
\]
Therefore, by the union bound, 
\begin{equation}
\label{eqn:unionboundstar}
\Pr[\mbox{There exist }i\in [k] \mbox{ and } \bfx\in X_i \mbox{ such that }\hdist(\bfx,\bfc_i)>d'-d_{\bfx}+\epsilon m] \leq n\cdot e^{-\frac{1}{3}m\epsilon^2}
\end{equation}
We remind that $m\geq 9c^2\log n/\delta^2 = 9\log n/\epsilon^2$ and so $n\cdot e^{-\frac{1}{3}m\epsilon^2} \le n^{-2}$. Thus, by \eqref{eqn:unionboundstar}, 
\begin{equation}
\label{eqn:case1}
\Pr[\mbox{There exist }i\in [k] \mbox{ and } \bfx\in X_i \mbox{ such that }\hdist(\bfx,\bfc_i)>d'-d_{\bfx}+\epsilon m] \leq n^{-2}.
\end{equation}
Since $d'\leq \opt(J')$ and $\opt(J')\geq m/c$, $d' + \epsilon m \le (1 + c \epsilon) \opt(J')$. Then, the probability that 
there exist $i\in [k]$ and  $\bfx\in X_i$ such that $\hdist(\bfx,\bfc_i)>(1+c\epsilon)\opt(J')-d_{\bfx}$ is at most $n^{-2}$ by \eqref{eqn:case1}.
Since $c\epsilon=\delta$, the proof is complete. 
\end{proof}


\section{Applications}
\label{sec:applications}

%
%

In this section we explain the impact of  Theorem~\ref{thm:kclustering} about  \rcenter to other  problems around low-rank matrix approximation.
We would like to mention 
that \rcenter is very similar to the \rclustering problem from \cite{DBLP:journals/corr/abs-1807-07156}. In \rcenter we want to minimize the maximum distance of a vector from the input set of vectors to the closest center, whereas in \rclustering the sum of distances is minimized.  While these problems are different, the reduction we explain here,   except a few details, are identical  to the ones described in \cite{DBLP:journals/corr/abs-1807-07156}. 
For reader's convenience,  we   give one reduction (Lemma~\ref{lem:matrixFas}) in   full details and skip all other reductions, which are similar. 


%

In the following lemma we show that 
\LRAGF\ is a special case of \rcenter.

\begin{lemma}\label{lem:matrixFas} 
There is an algorithm that given an  instance $(\bfA,r)$ of  \LRAGF, 
where $\bfA$ is an $m\times n$-matrix  and $r$ is an integer,
runs in time $\OO(m+n+2^{2r})$, and outputs an instance $J=(X,k=2^r,\cR)$ of \rcenter\ 
with the following property. 
Given any $\alpha$-approximate solution $C$ to $J$,  an $\alpha$-approximate solution $\bfB$ to $(\bfA,r)$ can be constructed in time $\OO(rmn)$ and vice versa.
%
%
\end{lemma}


\begin{proof}
Notice that if $\GFrank(\bfB)\leq r$, then $\bfB$ has at most $2^r$  distinct columns, because each column is a linear combination of at most $r$ vectors of a basis of the column space of $\bfB$.  Moreover, \LRAGF{} can also be stated as follows: find vectors $\bfs_1,\ldots,\bfs_r\in\{0,1\}^m$ such that 
$\max_{i\in[n]} \hdist(\bfa_i,S)$ 
is minimum, where
$\bfa_1,\ldots,\bfa_n$ are the columns of $\bfA$ and $S=\{\bfs \in \{0,1\}^m \colon \bfs~\text{is a linear combination of}~\bfs_1,\ldots,\bfs_r~\text{over}~\GF{}\}$. 

To  encode an instance of \LRAGF\ as an instance of \rcenter, 
we construct the following relation $R$. Set $k=2^r$.
Let $\Lambda=(\mathbf{\lambda}_1,\ldots,\mathbf{\lambda}_k)$ be the $k$-tuple composed of all 
distinct vectors in $\{0,1\}^r$.  Thus, each element $\lambda_i\in \Lambda$ is a binary $r$-vector. We define 
 $R=\{(x^\intercal \mathbf{\lambda}_1,\ldots,x^\intercal\mathbf{\lambda}_k)\mid x\in\{0,1\}^r\}.$  Thus, $R$ consists of $k=2^r$ $k$-tuples and every $k$-tuple in $R$ is a row of the matrix $\Lambda^\intercal \cdot \Lambda$.
 Now we define $X$ to be the set of columns of $\bfA$ and for each $i\in [m]$, $R_i=R$. Our algorithm outputs the instance $J=(X,k,\RR=(R_1,\ldots,R_m))$.

To show that the instance $(\bfA,r)$ of \LRAGF{} is equivalent to the constructed instance $J$, assume first that the vectors $\bfs_1,\ldots,\bfs_r\in \{0,1\}^m$
compose an (approximate) solution of  \LRAGF. 
For every $i\in[k]$  define the vector
$$\bfc_i=\lambda_i[1] \bfs_1\oplus\cdots\oplus\lambda_i[r] \bfs_r,$$
where $\mathbf{\lambda}_i^\intercal=(\lambda_i[1],\ldots,\lambda_i[r])$, $\oplus$ denotes the sum over \GF, and define the tuple $C=(\bfc_1,\ldots,\bfc_k)$. That is, $C$ contains all linear combinations of $\bfs_1,\ldots,\bfs_r$.
For every $i\in[k]$ and $j\in[m]$, we have that $\bfc_i[j]=(\bfs_1[j],\ldots,\bfs_r[j])\mathbf{\lambda}_i$. Therefore, $(\bfc_1[j],\ldots,\bfc_k[j])\in R$ for all $j\in[m]$. Thus, $C$ is a solution to $J$ of cost $\max_{\bfx\in X}\hdist(\bfx,C)$. 
 

For the opposite direction, assume that $C=(\bfc_1,\ldots,\bfc_k)$ is an (approximate) solution to $J$. We construct the vectors $\bfs_1,\ldots,\bfs_r$ as follows. 
Let $j\in[m]$. We have that $(\bfc_1[j],\ldots,\bfc_k[j])\in R$. Therefore, there is $\bfx\in\{0,1\}^r$ such that 
  $(\bfc_1[j],\ldots,\bfc_k[j])=(\bfx^\intercal\mathbf{\lambda}_1,\ldots,\bfx^\intercal\mathbf{\lambda}_k)$. 
  We set $\bfs_i[j]=\bfx[i]$ for $i\in[r]$. 
Observe that vectors in $C$ are linear combinations of the vectors $\bfs_1,\ldots,\bfs_r$. 
This immediately implies that for any $\alpha$-approximate solution $C$ of $J$ an $\alpha$-approximate solution $\bfB$ of $(\bfA,r)$ can be constructed in time $\OO(rmn)$.
%
%
\end{proof}

Thus, Theorem~\ref{thm:norm1PTAS} follows from Theorem~\ref{thm:kclustering} and Lemma~\ref{lem:matrixFas}.



\paragraph{Low Boolean-Rank Approximation.}
Let $\bfA$ be a binary $m\times n$ matrix. Now we consider the elements of $\bfA$ to be \emph{Boolean} variables. 
The \emph{Boolean rank} of $\bfA$ is the minimum $r$ such that $\bfA=\bfU\wedge \bfV$ for a Boolean $m\times r$ matrix $\bfU$ and a Boolean $r\times n$ matrix $\bfV$, where the product is Boolean, that is,  the logical $\wedge$ plays the role of multiplication and $\vee$ the role of sum. Here  $0\wedge 0=0$, $0 \wedge 1=0$, $1\wedge 1=1$ , $0\vee0=0$, $0\vee1=1$, and  $1\vee 1=1$.  
Thus the  matrix product is over the Boolean semi-ring $({0, 1}, \wedge, \vee)$. This can be equivalently expressed
as the  normal matrix product with addition defined as $1 + 1 =1$. Binary matrices equipped with such algebra are called \emph{Boolean
matrices}. 

In \LRAB, we are given an  $m\times n$ binary data matrix $\bfA$ and a positive integer $r$, and we seek a binary matrix $\bfB$ optimizing 
\begin{eqnarray}\label{eq_PCA_1111}
\text{ minimize } \|\bfA-\bfB\|_1  \nonumber \\ 
\text{ subject to } \rank(\bfB) \leq  r. \nonumber
\end{eqnarray}
Here, by the rank of binary matrix $\bfB$ we mean  its  Boolean rank, and norm $\| \cdot\|_1$ is   
the \emph{column sum norm}.  Similar to Lemma~\ref{lem:matrixFas}, one can prove that \LRAB is a special case of \rcenter, where $k = 2^r$. Thus, we get the following corollary from  Theorem~\ref{thm:kclustering}.

\begin{corollary}
\label{cor:booleancase}
There is an algorithm for \LRAB  that given an instance $I=(\bfA,r)$ and $0<\varepsilon<1$, runs in time 
$m^{\OO(1)}n^{\OO(2^{4r}/\varepsilon^4)}$, 
and outputs a $(1+\varepsilon)$-approximate solution with probability at least $1-2n^{-2}$. 
\end{corollary}

\paragraph*{Projective $k$-center.}  
The {\sc Binary Projective $k$-Center} problem is a variation of the {\sc Binary $k$-Center} problem, where 
the centers of clusters are linear subspaces of bounded dimension $r$. (For $r=1$ this is  {\sc Binary $k$-Center} and for $k=1$ this is \LRAGF.)
Formally, in 
{\sc Binary Projective $k$-Center} we are given a set $X\subseteq \{0,1\}^m$ of $n$ vectors and  positive integers $k$ and $r$. The objective is to find a family of $r$-dimensional linear subspaces   $C=\{C_1,\ldots,C_k\} $  over \GF minimizing  \(\max_{\bfx\in X} \hdist(\bfx,\bigcup_{i=1}^kC).\)

%

To see that {\sc Binary Projective $k$-Center} is a special case of  \rcenter, we observe that the condition that $C_i$ is an $r$-dimensional subspace over \GF can be encoded (as in Lemma~\ref{lem:matrixFas}) by $2^r$ constraints.  This observation leads to the following lemma. 

\begin{lemma}\label{lem:proj} 
There is an algorithm that given an  instance $(X,r,k)$ of  {\sc Binary Projective $k$-Center}, 
runs in time $\OO(m+n+2^{\OO(rk)})$, and outputs an instance $J=(X,k'=2^{kr},\cR)$ of \rcenter\ 
with the following property. 
Given any $\alpha$-approximate solution $C$ to $J$,  an $\alpha$-approximate solution $C'$ to $(X,r,k)$ can be constructed in time $\OO(rkmn)$ and vice versa.
\end{lemma}

Combining Theorem~\ref{thm:kclustering} and Lemma~\ref{lem:proj} together, we get the following corollary. 

\begin{corollary}
\label{cor:booleancase}
There is an algorithm for {\sc Binary Projective $k$-Center} that given an instance $I=(X,r,k)$ and $0<\varepsilon<1$, where $X\subseteq \{0,1\}^m$ is a set of $n$ vectors and $r,k\in {\mathbb N}$, runs in time 
$m^{\OO(1)}n^{\OO(2^{4kr}/\varepsilon^4)}$, 
and outputs a $(1+\varepsilon)$-approximate solution with probability at least $1-2n^{-2}$. 
\end{corollary}

\bibliographystyle{siam}
\bibliography{book_pc,pca_with_outliers,k-clustering}

\begin{thebibliography}{10}

\bibitem{BanBBKLW19}
{\sc F.~Ban, V.~Bhattiprolu, K.~Bringmann, P.~Kolev, E.~Lee, and D.~P.
  Woodruff}, {\em A {PTAS} for {\(\mathscr{l}\)}p-low rank approximation}, in
  Proceedings of the Thirtieth Annual {ACM-SIAM} Symposium on Discrete
  Algorithms, {SODA} 2019, San Diego, California, USA, January 6-9, 2019,
  {SIAM}, 2019, pp.~747--766.

\bibitem{Bartl2010}
{\sc E.~Bartl, R.~Belohl{\'{a}}vek, and J.~Konecny}, {\em Optimal
  decompositions of matrices with grades into binary and graded matrices},
  Annals of Mathematics and Artificial Intelligence, 59 (2010), pp.~151--167.

\bibitem{BelohlavekV10}
{\sc R.~Belohl{\'{a}}vek and V.~Vychodil}, {\em Discovery of optimal factors in
  binary data via a novel method of matrix decomposition}, J. Computer and
  System Sciences, 76 (2010), pp.~3--20.

\bibitem{berry1995using}
{\sc M.~W. Berry, S.~T. Dumais, and G.~W. O'Brien}, {\em Using linear algebra
  for intelligent information retrieval}, SIAM review, 37 (1995), pp.~573--595.

\bibitem{BlumHK17}
{\sc A.~Blum, J.~Hopcroft, and R.~Kannan}, {\em Foundations of Data Science},
  June 2017.

\bibitem{cygan_et_al:LIPIcs:2016:6023}
{\sc M.~Cygan, D.~Lokshtanov, M.~Pilipczuk, M.~Pilipczuk, and S.~Saurabh}, {\em
  {Lower Bounds for Approximation Schemes for Closest String}}, in Proceedings
  of the 15th Scandinavian Symposium and Workshops on Algorithm Theory (SWAT),
  vol.~53 of Leibniz International Proceedings in Informatics (LIPIcs),
  Dagstuhl, Germany, 2016, Schloss Dagstuhl--Leibniz-Zentrum fuer Informatik,
  pp.~12:1--12:10.

\bibitem{DanHJWZ15}
{\sc C.~Dan, K.~A. Hansen, H.~Jiang, L.~Wang, and Y.~Zhou}, {\em On low rank
  approximation of binary matrices}, CoRR, abs/1511.01699 (2015).

\bibitem{eckart1936approximation}
{\sc C.~Eckart and G.~Young}, {\em The approximation of one matrix by another
  of lower rank}, Psychometrika, 1 (1936), pp.~211--218.

\bibitem{DBLP:journals/corr/abs-1807-07156}
{\sc F.~V. Fomin, P.~A. Golovach, D.~Lokshtanov, F.~Panolan, and S.~Saurabh},
  {\em Approximation schemes for low-rank binary matrix approximation
  problems}, CoRR, abs/1807.07156 (2018).

\bibitem{FominGP18}
{\sc F.~V. Fomin, P.~A. Golovach, and F.~Panolan}, {\em Parameterized low-rank
  binary matrix approximation}, in 45th International Colloquium on Automata,
  Languages, and Programming, {ICALP} 2018, July 9-13, 2018, Prague, Czech
  Republic, vol.~107 of LIPIcs, Schloss Dagstuhl - Leibniz-Zentrum fuer
  Informatik, 2018, pp.~53:1--53:16.

\bibitem{fu2010binary}
{\sc Y.~Fu, N.~Jiang, and H.~Sun}, {\em Binary matrix factorization and
  consensus algorithms}, in Proceedings of the International Conference on
  Electrical and Control Engineering (ICECE), IEEE, 2010, pp.~4563--4567.

\bibitem{GasieniecJL04}
{\sc L.~Gasieniec, J.~Jansson, and A.~Lingas}, {\em Approximation algorithms
  for hamming clustering problems}, J. Discrete Algorithms, 2 (2004),
  pp.~289--301.

\bibitem{GillisV15}
{\sc N.~Gillis and S.~A. Vavasis}, {\em On the complexity of robust {PCA} and
  \mbox{$\ell_1$}-norm low-rank matrix approximation}, CoRR, abs/1509.09236
  (2015).

\bibitem{GutchGYT12}
{\sc H.~W. Gutch, P.~Gruber, A.~Yeredor, and F.~J. Theis}, {\em {ICA} over
  finite fields - separability and algorithms}, Signal Processing, 92 (2012),
  pp.~1796--1808.

\bibitem{hotelling1933analysis}
{\sc H.~Hotelling}, {\em Analysis of a complex of statistical variables into
  principal components.}, Journal of educational psychology, 24 (1933), p.~417.

\bibitem{DBLP:conf/icdm/JiangH13a}
{\sc P.~Jiang and M.~T. Heath}, {\em Mining discrete patterns via binary matrix
  factorization}, in {ICDM} Workshops, {IEEE} Computer Society, 2013,
  pp.~1129--1136.

\bibitem{Jiang2014}
{\sc P.~Jiang, J.~Peng, M.~Heath, and R.~Yang}, {\em A clustering approach to
  constrained binary matrix factorization}, in Data Mining and Knowledge
  Discovery for Big Data: Methodologies, Challenge and Opportunities, Springer
  Berlin Heidelberg, Berlin, Heidelberg, 2014, pp.~281--303.

\bibitem{JiaoXL04}
{\sc Y.~Jiao, J.~Xu, and M.~Li}, {\em On the $k$-closest substring and
  $k$-consensus pattern problems}, in Proceedings of the 15th Annual Symposium
  on Combinatorial Pattern (CPM), vol.~3109 of Lecture Notes in Comput. Sci.,
  Springer, 2004, pp.~130--144.

\bibitem{DBLP:journals/jacm/Kleinberg99}
{\sc J.~M. Kleinberg}, {\em Authoritative sources in a hyperlinked
  environment}, J. {ACM}, 46 (1999), pp.~604--632.

\bibitem{Koyuturk2003}
{\sc M.~Koyut\"{u}rk and A.~Grama}, {\em Proximus: A framework for analyzing
  very high dimensional discrete-attributed datasets}, in Proceedings of the
  9th {ACM SIGKDD} International Conference on Knowledge Discovery and Data
  Mining (KDD), New York, NY, USA, 2003, ACM, pp.~147--156.

\bibitem{LiMW02}
{\sc M.~Li, B.~Ma, and L.~Wang}, {\em On the closest string and substring
  problems}, J. {ACM}, 49 (2002), pp.~157--171.

\bibitem{LuVAH12}
{\sc H.~Lu, J.~Vaidya, V.~Atluri, and Y.~Hong}, {\em Constraint-aware role
  mining via extended boolean matrix decomposition}, {IEEE} Trans. Dependable
  Sec. Comput., 9 (2012), pp.~655--669.

\bibitem{MaS09}
{\sc B.~Ma and X.~Sun}, {\em More efficient algorithms for closest string and
  substring problems}, SIAM J. Computing, 39 (2009), pp.~1432--1443.

\bibitem{MahonyM11}
{\sc M.~W. Mahoney}, {\em Randomized algorithms for matrices and data},
  Foundations and Trends in Machine Learning, 3 (2011), pp.~123--224.

\bibitem{MiettinenMGDM08}
{\sc P.~Miettinen, T.~Mielik{\"{a}}inen, A.~Gionis, G.~Das, and H.~Mannila},
  {\em The discrete basis problem}, {IEEE} Trans. Knowl. Data Eng., 20 (2008),
  pp.~1348--1362.

\bibitem{DBLP:conf/kdd/MiettinenV11}
{\sc P.~Miettinen and J.~Vreeken}, {\em Model order selection for boolean
  matrix factorization}, in Proceedings of the 17th {ACM} {SIGKDD}
  International Conference on Knowledge Discovery and Data Mining (KDD), {ACM},
  2011, pp.~51--59.

\bibitem{MR0114821}
{\sc L.~Mirsky}, {\em Symmetric gauge functions and unitarily invariant norms},
  Quart. J. Math. Oxford Ser. (2), 11 (1960), pp.~50--59.

\bibitem{Mitra:2016}
{\sc B.~Mitra, S.~Sural, J.~Vaidya, and V.~Atluri}, {\em A survey of role
  mining}, ACM Comput. Surv., 48 (2016), pp.~50:1--50:37.

\bibitem{OstrovskyR02}
{\sc R.~Ostrovsky and Y.~Rabani}, {\em Polynomial-time approximation schemes
  for geometric min-sum median clustering}, J. ACM, 49 (2002), pp.~139--156.

\bibitem{PainskyRF16}
{\sc A.~Painsky, S.~Rosset, and M.~Feder}, {\em Generalized independent
  component analysis over finite alphabets}, {IEEE} Trans. Information Theory,
  62 (2016), pp.~1038--1053.

\bibitem{pearson1901liii}
{\sc K.~Pearson}, {\em Liii. on lines and planes of closest fit to systems of
  points in space}, The London, Edinburgh, and Dublin Philosophical Magazine
  and Journal of Science, 2 (1901), pp.~559--572.

\bibitem{Shen2009}
{\sc B.-H. Shen, S.~Ji, and J.~Ye}, {\em Mining discrete patterns via binary
  matrix factorization}, in Proceedings of the 15th {ACM SIGKDD} International
  Conference on Knowledge Discovery and Data Mining (KDD), New York, NY, USA,
  2009, ACM, pp.~757--766.

\bibitem{SongWZ17}
{\sc Z.~Song, D.~P. Woodruff, and P.~Zhong}, {\em Low rank approximation with
  entrywise $\ell_1$-norm error}, in Proceedings of the 49th Annual {ACM}
  {SIGACT} Symposium on Theory of Computing (STOC), {ACM}, 2017, pp.~688--701.

\bibitem{DBLP:conf/icde/Vaidya12}
{\sc J.~Vaidya}, {\em Boolean matrix decomposition problem: Theory, variations
  and applications to data engineering}, in Proceedings of the 28th {IEEE}
  International Conference on Data Engineering (ICDE), {IEEE} Computer Society,
  2012, pp.~1222--1224.

\bibitem{VaidyaAG07}
{\sc J.~Vaidya, V.~Atluri, and Q.~Guo}, {\em The role mining problem: finding a
  minimal descriptive set of roles}, in Proceedings of the 12th {ACM} Symposium
  on Access Control Models and (SACMAT), 2007, pp.~175--184.

\bibitem{Woodr14}
{\sc D.~P. Woodruff}, {\em Sketching as a tool for numerical linear algebra},
  Foundations and Trends in Theoretical Computer Science, 10 (2014),
  pp.~1--157.

\bibitem{Yeredor11}
{\sc A.~Yeredor}, {\em Independent component analysis over {G}alois fields of
  prime order}, {IEEE} Trans. Information Theory, 57 (2011), pp.~5342--5359.

\end{thebibliography}

\end{document}